\pdfoutput=1
\documentclass{llncs}
\usepackage[utf8]{inputenc}

\usepackage[ruled,vlined,linesnumbered]{algorithm2e}
\usepackage{mathtools}
\usepackage{thm-restate}
\usepackage{hyperref}
\usepackage{cleveref}
\usepackage{thmtools}

\SetKwComment{tcc}{$\triangleright$~}{}
\SetKwComment{tcp}{$\triangleright$~}{}
\SetCommentSty{normalfont}

\declaretheorem[style=theorem,name=Observation,refname={observation,observations},Refname={Observation,Observations}]{myobservation}

\newcommand\msa{\textsf{MSA}\xspace}
\newcommand\msas{\textsf{MSA}s\xspace}
\newcommand\msamn{\textsf{MSA}$[1..m,1..n]$\xspace}

\newcommand\edses{\textsf{EDS}es\xspace}
\newcommand\efg{\textsf{EFG}\xspace}
\newcommand\efgs{\textsf{EFG}s\xspace}

\newcommand\msaij[2]{\mathsf{MSA}[#1,#2]\xspace}
\newcommand{\gap}{\mathtt{-}}
\newcommand\spell{\mathsf{spell}\xspace}
\newcommand\gst{$\mathsf{GST}_\mathsf{MSA}$\xspace}

\newcommand\atrue{\mathrm{\mathbf{true}}\xspace}
\newcommand\afalse{\mathrm{\mathbf{false}}\xspace}
\newcommand{\leftmostleaf}{\mathrm{leftmostleaf}}
\newcommand{\rightmostleaf}{\mathrm{rightmostleaf}}
\newcommand{\nextleaf}{\mathrm{nextleaf}}
\newcommand{\prevleaf}{\mathrm{prevleaf}}
\newcommand{\suffixlink}{\mathrm{suffixlink}}

\DeclareMathOperator{\cchar}{char}
\DeclareMathOperator{\sstring}{string}
\DeclareMathOperator{\parent}{parent}
\DeclareMathOperator{\stringdepth}{stringdepth}

\title{Linear Time Construction of \\ Indexable Elastic Founder Graphs}
\author{Nicola Rizzo\orcidID{0000-0002-2035-6309} \and Veli Mäkinen\orcidID{0000-0003-4454-1493}}
\institute{Department of Computer Science, University of Helsinki, Finland \\ \email{\{nicola.rizzo,veli.makinen\}@helsinki.fi}}
\date{January 2022}

\begin{document}

\maketitle

\begin{abstract}
Pattern matching on graphs has been widely studied lately due to its importance in genomics applications. Unfortunately, even the simplest problem of deciding if a string appears as a subpath of a graph admits  a quadratic lower bound under the Orthogonal Vectors Hypothesis (Equi et al. ICALP 2019, SOFSEM 2021).
To avoid this bottleneck, the research has shifted towards more specific graph classes, e.g. those induced from multiple sequence alignments (\msas).
Consider segmenting $\msa[1..m,1..n]$ into $b$ blocks
$\msa[1..m,1..j_1]$, $\msa[1..m,j_1+1..j_2]$, $\ldots$, $\msa[1..m,j_{b-1}+1..n]$.
The distinct strings in the rows of the blocks, after the removal of gap symbols, form the nodes of an \emph{elastic founder graph} (\efg) where the edges represent the original connections observed in the $\msa$.
An \efg is called \emph{indexable} if a node label occurs as a prefix of only those paths that start from a node of the same block.
Equi et al.\ (ISAAC 2021) showed that such \efgs support fast pattern matching
and gave an $O(mn \log m)$-time algorithm for preprocessing the \msa in a way that allows the construction of indexable \efgs maximizing the number of blocks and, alternatively, minimizing the maximum length of a block, in $O(n)$ and $O(n \log\log n)$ time respectively.
Using the suffix tree and solving a novel ancestor problem on trees, we improve the preprocessing to $O(mn)$ time and the $O(n \log \log n)$-time \efg construction to $O(n)$ time, thus showing that both types of indexable \efgs can be constructed in time linear in the input size.
\keywords{multiple sequence alignment \and pattern matching \and data structures \and segmentation algorithms \and dynamic programming \and suffix tree}
\end{abstract}

\subsubsection*{Acknowledgements}
This project has received funding from the European Union’s Horizon
2020 research and innovation programme under the Marie \\ Sk\l{}odowska-Curie
grant agreement No 956229.

\section{Introduction}

Searching strings in a graph has become a central problem along with the development of high-throughput sequencing techniques. Namely, thousands of human genomes are now available, forming a so-called \emph{pangenome} of a species \cite{marschall2016computational}. Such pangenome can be used for enhancing various analysis tasks that have previously been conducted with a single reference genome \cite{MNSV09jcb,Sch09,SVM14,Garetal18,hisat2,graphtyper2,Norrietal21}.
The most popular representation for a pangenome is a graph, whose paths spell the input genomes. The basic primitive required on such pangenome graphs is to be able to search occurrences of query strings (short reads) as subpaths of the graph. Unfortunately, even finding exact matches of a query string of length $q$ in a graph with $e$ edges cannot be done significantly faster than $O(qe)$ time unless the Orthogonal Vectors Hypothesis (OVH) is false \cite{EGMT19}. Therefore, practical tools deploy various heuristics or use other pangenome representations as a basis. 

\begin{figure}
    \centering
    \includegraphics{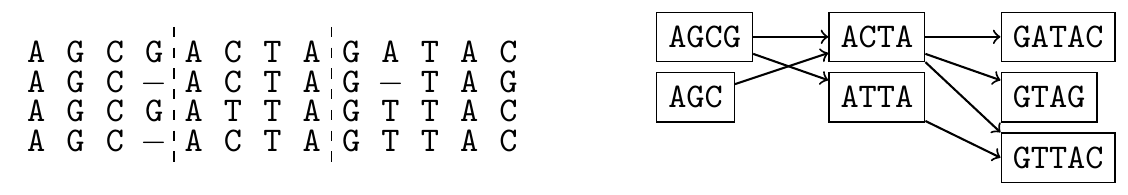}
    \caption{An indexable elastic founder graph induced from a segmentation of an \msa.
    The example is adapted from Equi et al.~\cite{Equietal21}.}\label{fig:segmentation}
\end{figure}

Due to the difficulty of string search in general graphs, M\"akinen et al.~\cite{MCENT20} and Equi et al.~\cite{Equietal21} studied graphs induced from multiple sequence alignments (\msas), as we describe in \Cref{sect:definitions}. Any segmentation of an \msa naturally induces a graph consisting of nodes partitioned into blocks with edges connecting consecutive blocks. Such \emph{elastic founder graph} (\efg) is illustrated in Figure~\ref{fig:segmentation}. The key observation is that if the resulting node labels do not appear as a prefix of any other path than those starting at the same block, then there is an index structure for the graph that supports fast pattern matching \cite{MCENT20,Equietal21}. Equi et al.~\cite{Equietal21} also show that some such indexability property is required as the OVH-based lower bound holds even on these induced graphs. M\"akinen et al.~\cite{MCENT20} gave an $O(mn)$ time algorithm to construct an indexable \efg with minimum maximum block length, given a gapless \msamn. Equi et al.~\cite{Equietal21} extended the result to general \msas. They obtained an $O(mn \log m)$-time preprocessing algorithm which allows the construction of indexable \efgs maximizing the number of blocks and, alternatively, minimizing the maximum length of a block, in $O(n)$ and in $O(n \log\log n)$ time, respectively. We recall these results in \Cref{sect:efgconstruction}, and we refer the reader to the aforementioned papers for connections of the approach to elastic degenerate strings and Wheeler graphs.  

In this paper, we improve the preprocessing algorithm of Equi et al.\ to $O(mn)$ by performing an in-depth analysis of their solution based on the generalized suffix tree \gst built from the gaps-removed rows of the \msa (\Cref{sect:preprocessing}). Even though removing gaps constitutes a loss of essential information, this information can be fed back into the structure by considering the right subsets of its nodes or leaves. Then, the main step in preprocessing the \msa is solving a novel ancestor problem on the tree structure of \gst that we call the \emph{exclusive ancestor set problem}, and as our main contribution, we identify such problem and provide a linear-time solution.
This result directly improves the solution by Equi et al.\ for constructing indexable \efgs maximizing the number of blocks from $O(mn\log m)$ to $O(mn)$ time. Moreover, in \Cref{sect:minmaxlength} we give a new algorithm that after the $O(mn)$-time preprocessing can construct indexable \efgs minimizing the maximum block length in $O(n)$ time. Hence, we show that both type of indexable elastic founder graphs can be constructed in time linear in the input size.

\section{Definitions\label{sect:definitions}}

We follow the notation of Equi et al.~\cite{Equietal21}. 

\paragraph{Strings.}
We denote integer intervals by $[x..y]$. Let $\Sigma = [1..\sigma]$ be an alphabet of size $\lvert \Sigma \rvert = \sigma$.
A \emph{string} $T[1..n]$ is a sequence of symbols from $\Sigma$, i.e. $T\in \Sigma^n$, where 
$\Sigma^n$ denotes the set of strings of length $n$ over $\Sigma$.
A \emph{suffix} (\emph{prefix}) of string $T[1..n]$ is $T[i..n]$ ($T[1..i]$) for $1\leq i\leq n$ ($1\leq i\leq n$) and we say it is \emph{proper} if $i > 1$ ($i < n$).
The \emph{length} of a string $T$ is denoted $|T|$ and the \emph{empty string} $\varepsilon$ is the string of length $0$.
In particular, substring $T[i..j]$ where $j<i$ is the empty string.
The \emph{lexicographic order} of two strings $A$ and $B$ is naturally defined by the order of the alphabet: $A<B$ iff $A[1..i]=B[1..i]$ and $A[i+1]<B[i+1]$ for some $i\geq 0$. If $i+1>\min(|A|,|B|)$, then the shorter one is regarded as smaller. However, we usually avoid this implicit comparison by adding an \emph{end marker} $\$$ to the strings.
Concatenation of strings $A$ and $B$ is denoted $A \cdot B$, or just $AB$.

\paragraph{Elastic founder graphs.}

\msas can be compactly represented by elastic founder graphs, i.e.\ the vertex-labeled graphs that we formalize in this section.

A \emph{multiple sequence alignment} \msamn is a matrix with $m$ strings drawn from $\Sigma \cup \{\gap\}$, each of length $n$, as its rows. Here $\gap \notin \Sigma$ is the \emph{gap} symbol. For a string $X \in \left(\Sigma \cup \{\gap\}\right)^*$, we denote $\spell(X)$ the string resulting from removing the gap symbols from $X$. If an \msa does not contain gaps then we say it is \emph{gapless}, otherwise we say that it is a \emph{general} \msa.

Let $\mathcal{P}$ be a \emph{partitioning} of $[1..n]$, that is, a sequence of subintervals $\mathcal{P}=[x_1..y_1]$, $[x_2..y_2],\ldots,[x_b..y_b]$, where $x_1=1$, $y_b=n$, and for all $j>2$, $x_j=y_{j-1}+1$. A \emph{segmentation} $S$ of \msamn based on partitioning $\mathcal{P}$ is the sequence of $b$ sets $S^k= \{\spell(\msaij{i}{x_k..y_k}) \mid 1\leq i\leq m\}$ for $1\leq k\leq b$; in addition, we require for a (proper) segmentation that $\spell(\msaij{i}{x_k..y_k})$ is not an empty string for any $i$ and $k$. We call set $S^k$ a \emph{block}, while $\msaij{1..m}{x_k..y_k}$ or just $[x_k..y_k]$ is called a \emph{segment}. The \emph{length} of block $S^k$ is $L(S^k)=y_k-x_k+1$.

Segmentation naturally leads to the definition of a founder graph through the block graph concept.
\begin{definition}[Block Graph]
A \emph{block graph} is a graph $G=(V,E,\ell)$ where $\ell: V \rightarrow \Sigma^+$ is a function that assigns a string label to every node and for which the following properties hold:
\begin{enumerate}
	\item set $V$ can be partitioned into a sequence of $b$ \emph{blocks} $V^1, V^2, \ldots, V^b$, that is, $V = V^1 \cup V^2 \cup \cdots \cup V^b$ and $V^i \cap V^j = \emptyset$ for all $i\neq j$;
	\item if $(v,w) \in E$ then $v \in V^i$ and $w \in V^{i+1}$ for some $1 \leq i \leq b-1$; and
	\item if $v,w \in V^i$ then $|\ell(v)| = |\ell(w)|$ for each $1 \leq i \leq b$ and if $v\neq w$, $\ell(v) \neq \ell(w)$.
\end{enumerate}
\end{definition}
For \emph{gapless} \msas, block $S^k$ equals segment $\msaij{1..m}{x_k..y_k}$, and in that case the \emph{founder graph} is a block graph induced by segmentation $S$~\cite{MCENT20}.
The idea is to have a graph in which the nodes represent the strings in $S$ while the edges retain the information of how such strings can be recombined to spell any sequence in the original \msa.

For general \msas with gaps, we consider the following extension, with an analogy to elastic degenerate strings (\edses)~\cite{bernardini_et_al2019elastic}:
\begin{definition}[Elastic block and founder graphs]\label{def:efg}
We call a block graph \emph{elastic} if its third condition is relaxed in the sense that each $V^i$ can contain non-empty variable-length strings. An \emph{elastic founder graph} (\efg) is an elastic block graph $G(S) = (V,E,\ell)$ \emph{induced} by a segmentation $S$ as follows: for each $1 \leq k \leq b$ we have $S^k = \{\spell(\msaij{i}{x_k..y_k}) \mid 1\leq i\leq m\} = \{\ell(v) : v \in V^k\}$.
It holds that $(v,w) \in E$ if and only if there exists $k \in [1 .. b-1]$ and $i \in [1 .. m]$ such that $v \in V^k$, $w \in V^{k+1}$ and $\spell(\msaij{i}{x_k..y_{k+1}})= \ell(v)\ell(w)$.
\end{definition}
For example, in the general $\msaij{1..4}{1..13}$ of \Cref{fig:segmentation}, the segmentation based on partitioning $[1..4],[5..8],[9..13]$ induces an \efg $G(S) = (V^1 \cup V^2 \cup V^3, E, \ell)$ where the nodes in $V^1$ and $V^3$ have labels of variable length.

By definition, (elastic) founder and block graphs are acyclic. For convention, we interpret the direction of the edges as going from left to right. 
Consider a path $P$ in $G(S)$ between any two nodes. The label $\ell(P)$ of $P$ is the concatenation of labels of the nodes in the path. Let $Q$ be a query string. We say that $Q$ \emph{occurs} in $G(S)$ if $Q$ is a substring of $\ell(P)$ for any path $P$ of $G(S)$.

\begin{definition}[\cite{MCENT20}]\label{def:repeat-free}
\efg $G(S)$ is \emph{repeat-free} if each $\ell(v)$ for $v\in V$ occurs in $G(S)$ only as a prefix of paths starting with $v$. 
\end{definition}
\begin{definition}[\cite{MCENT20}]\label{def:semi-repeat-free}
\efg $G(S)$ is \emph{semi-repeat-free} if each $\ell(v)$ for $v\in V$ occurs in $G(S)$ only as a prefix of paths starting with $w\in V$, where $w$ is from the same block as $v$. 
\end{definition}
For example, the \efg of \Cref{fig:segmentation} is not repeat-free, since $\mathtt{AGC}$ occurs as a prefix of two distinct labels of nodes in the same block, but it is semi-repeat-free since all node labels $\ell(v)$ with $v \in V^k$ occur in $G(S)$ only starting from block $V^k$, or they do not occur at all elsewhere in the graph.
These definitions also apply to general elastic block graphs and to elastic degenerate strings as their special case.
We will discuss these two indexability properties together as the (semi-)repeat-free property, when applicable.

\paragraph{Basic tools.}
A \emph{trie}~\cite{Bri59} of a set of strings is a rooted directed tree with outgoing edges of each node labeled by distinct symbols such that there is a root-to-leaf path spelling each string in the set; the shared part of the root-to-leaf paths of two different leaves spell the common prefix of the corresponding strings. In a \emph{compact trie}, the maximal non-branching paths of a trie become edges labeled with the concatenation of labels on the path. The \emph{suffix tree} of $T \in \Sigma^*$ is the compact trie of all suffixes of string $T \$$. In this case, the edge labels are substrings of $T$ and can be represented in constant space as an interval. Such tree takes linear space and can be constructed in linear time~\cite{Farach97} so that when reading the leaves from left to right, the suffixes are listed in their lexicographic order. We say that two or more leaves of the suffix tree are \emph{adjacent} if they succeed one another when reading them left to right. A \emph{generalized suffix tree} is one built on a set of strings. In this case, string $T$ above is the concatenation of the strings with symbol $\$$ between each.

Let $Q[1..m]$ be a query string. If $Q$ occurs in $T$, then the \emph{locus} or \emph{implicit node} of $Q$ in the suffix tree of $T$ is $(v,k)$ such that $Q = XY$, where $X$ is the path spelled from the root to the parent of $v$ and $Y$ is the prefix of length $k$ of the edge from the parent of $v$ to $v$. The leaves of the subtree rooted at $v$, or \emph{the leaves covered by $v$}, are then all the suffixes sharing the common prefix $Q$.
Let $aX$ and $X$ be paths spelled from the root of a suffix tree to nodes $v$ and $w$, respectively. Then one can store a \emph{suffix link} from $v$ to $w$.

String $B[1..n]$ from a binary alphabet is called a \emph{bitvector}. Operation $\mathrm{rank}(B,i)$ returns the number of 1s in $B[1..i]$. Operation $\mathrm{select}(B,j)$ returns the index $i$ containing the $j$-th 1 in $B$. Both queries can be answered in constant time using an index requiring $o(n)$ bits in addition to the bitvector itself~\cite{Jac89}.

\section{Overview of \efg construction algorithms}\label{sect:efgconstruction}

Equi et al.\ have shown that (semi-)repeat-free \efgs are easy to index for fast pattern matching~\cite{Equietal21}, and as we describe in \Cref{sub:segmentation} they extended the previous research for the gapless and repeat-free setting showing that finding \mbox{(semi-)} repeat-free elastic founder graphs is equivalent to finding (semi-)repeat-free \msa segmentations.
Plus, to show that the (semi-)repeat-free property does not hinder the flexibility in choosing the resulting \efgs, they considered e.g.~the following score functions for \msa segmentations: $i$.\ maximizing the number of blocks; and $ii$.\ minimizing the maximum length of a block.

In the gapless and repeat-free setting, scores $i$.\ and $ii$.\ admit the construction of indexable founder graphs in $O(mn)$ time, thanks to previous research on founder graphs and \msa segmentations~\cite{MCENT20,NCKM19,CKMN19}.
In the general and semi-repeat-free setting, Equi et al.\ have given $O(mn \log m)$ and $O(mn \log m + n\log \log n)$-time algorithms for scores $i$.\ and $ii$., respectively, based on a common preprocessing of the \msa that we review in \Cref{sub:efgconstruction}.

\subsection{Segmentation characterization for indexable \efgs}\label{sub:segmentation}
Consider a segmentation $S = S^1, S^2, \ldots, S^b$ that induces a (semi-)repeat-free \efg $G(S)=(V,E,\ell)$, as per \Cref{def:efg}. The strings occurring in graph $G(S)$ are a superset of the strings occurring in the original \msa rows because each node label can represent \emph{multiple} rows and each edge $(v,w) \in E$ means the existence of \emph{some} row spelling $\ell(v)\ell(w)$ in the corresponding consecutive segments. For example, string $\mathtt{GACTAGT}$ occurs in the \efg of \Cref{fig:segmentation} but it does not occur in any row of the original \msa.

The (semi-)repeat-free property involves graph $G(S)$, but luckily it does not depend on the new strings added in the founder graph and it can be checked only against the \msa and segmentation $S$. This simplifies choosing a segmentation resulting in an indexable founder graph and it was initially proven by Mäkinen et al.\ in the gapless and repeat-free setting.

\begin{lemma}[Characterization, gapless setting~\cite{MCENT20}]
We say that a segment $[x,y]$ of a gapless $\msaij{1..m}{1..n}$ is repeat-free if string $\msaij{i}{x..y}$ occurs in the \msa only at position $x$ of some row, for some $1 \le i \le m$. Then $G(S)$ is repeat-free if and only if all segments of $S$ are repeat-free.
\end{lemma}
Equi et al.\ in \cite{Equietal21} refined this property for \msas with gaps but did not provide an explicit proof. Since it is essential to the correctness of the construction algorithms, we provide such a proof here.
\begin{restatable}[Characterization~\cite{Equietal21}]{lemma}{characterization}\label{lem:characterization}
We say that segment $[x,y]$ of a general $\msaij{1..m}{1..n}$ is semi-repeat-free if for any $i,i' \in [1..m]$ string $\spell(\msaij{i}{x..y})$ occurs in gaps-removed row $\spell(\msaij{i'}{1..n})$ only at position $g(i',x)$, where $g(i',x)$ is equal to $x$ subtracted the number of gaps in $\msaij{i'}{1..x}$. Similarly, $[x,y]$ is repeat-free if the eventual occurrence of $\spell(\msaij{i}{1..n})$ at position $g(i',x)$ in row $i'$ also ends at position $g(i',y)$. Then $G(S)$ is (semi-)repeat-free if and only if all segments of $S$ are (semi-)repeat-free.
\end{restatable}
\begin{proof}
For convenience, we say that a segment or a founder graph is \emph{valid} if it is (semi-)repeat-free, otherwise it is \emph{invalid}. Moreover, we define the following notion of a standard string occurring in $G(S) = (V,E,\ell)$.
We say that $S \in \Sigma^+$ is a \emph{standard substring of path $P = w_1, \dots, w_k$ in $G(S)$} if $P$ spells $S$ using all of its vertices, meaning
\[
    S = \ell(w_1) \big[ j..\lvert \ell(w_1) \rvert \big] \cdot \ell(w_2) \cdots \ell(w_{k-1}) \cdot \ell(w_k) \big[ 1..j' \big]
\]
with $1 \le j \le \lvert \ell(w_1) \rvert$ and $1 \le j' \le \lvert \ell(w_k) \rvert$.
We also say that the occurrence of $S$ through $P$ involves $k$ vertices of $G(S)$.

We carry out the proof of the two sides by proving their contrapositions and using the following facts:
\begin{enumerate}
    \item a segment $[x,y]$ is invalid if and only if there exist $i,i' \in [1..m]$ such that string $\spell(\msaij{i}{x..y})$ occurs in row $\spell(\msaij{i'}{1..n})$ at some position other than $g(i',x)$, or string $\spell(\msaij{i}{x..y})$ is a proper prefix of string $\spell(\msaij{i'}{x..y})$ (for the semi-repeat-free case ignore this last condition);
    \item founder graph $G(S)$ is invalid if and only if there exists node $v \in V$ such that $\ell(v)$ is a standard substring of some path $P = w_1, \dots, w_k$ in $G(S)$ and one of the following holds: $w_1$ is in a different block than $v$, $\ell(v)$ occurs in $\ell(P)$ at some position other than 1, or $k = 1$ and $\ell(v)$ is a proper prefix of $\ell(P) = \ell(w_1)$ (for the semi-repeat-free case, ignore this last condition).
\end{enumerate}
It is immediate to see that by construction of $G(S)$ the additional invalidity conditions exclusive to the repeat-free case (the last conditions of facts 1.\ and 2.) are equivalent, so we concentrate on the conditions in common with the semi-repeat-free case.

($\Rightarrow$) Let $[x,y]$ be an invalid segment of $S$, with string $\spell(\msaij{i}{x..y})$ occurring in row $i'$ at some position $j$ other than $g(i',x)$, for some $i,i' \in [ 1 .. m ]$, and let $v \in V$ be the node in the block corresponding to segment $[x,y]$ such that $\ell(v) = \spell(\msaij{i}{x..y})$. If $g(i',x) < j \le g(i',y)$ then $\ell(v)$ occurs in $\ell(P)$ at position $p - g(i',x) \neq 1$, with $P$ a path starting from the same block of $v$, otherwise $j < g(i',x)$ or $j > g(i',y)$ and $\ell(v)$ occurs in some path of $G(S)$ starting from a node in a different block than that of $v$. In both cases $G(S)$ is invalid.

($\Leftarrow$) If $G(S)$ is invalid, let $\ell(v)$ be a standard substring of some path $P = w_1, \dots, w_k$ of $G(S)$ making the founder graph invalid, for some $v \in V$. Following the same arguments as in~\cite[Section 5.1]{MCENT20}, if $k \le 2$ then $\ell(P)$ is a substring of some row of the input \msa that makes $S$ invalid, since by construction of $G(S)$ for every edge $(u,u') \in E$ it holds that $\ell(w)\ell(w')$ occurs in the \msa.
Otherwise $k > 2$, i.e.\ the occurrence of $\ell(v)$ through $P$ involves at least three vertices, and $\ell(v) = A \ell(w) B$ for some $w \in \lbrace w_2, \dots, w_{k-1} \rbrace \subseteq V$,  $A,B \in \Sigma^+$. But then $\ell(w) \in \Sigma^+$ occurs in $\ell(v)$ at some position other than 1 and so there are row indices $i,i' \in [1..m]$ such that $\ell(w) = \spell(\msaij{i}{x..y})$ occurs in $\spell(\msaij{i'}{1..n})$ at some position other than $g(i',x)$, where $[x,y]$ is the segment of $S$ corresponding to the block of $w$ and $\spell(\msaij{i'}{1..n})$ contains $\ell(v)$, making segment $[x,y]$ invalid.
\end{proof}

\subsection{\efg construction algorithms}\label{sub:efgconstruction}
Just as in the gapless and repeat-free setting, \Cref{lem:characterization} implies that the optimal score $s(j)$ of a (semi-)repeat-free segmentation of the general \msa prefix $\msaij{1..m}{1..j}$ can be computed recursively for a variety of scoring schemes:
\begin{equation}\label{eq:score}
    s(j) =
    \bigoplus_{\substack{j' \,:\, 0 \le j' < j \;\text{s.t.} \\ \msaij{1..m}{j'+1..j} \,\text{is} \\ \text{(semi-)repeat-free}}}
    g \big(
    s(j'), j', j
    \big)
\end{equation}
where operator $\bigoplus$ and function $g$ depend on the desired scoring scheme. Indeed:
\begin{itemize}
    \item[$i$.]
for $s(j)$ to be equal to the optimal score of a segmentation maximizing the number of blocks, set $\bigoplus = \max$ and $g(s(j'),j',j) = s(j') + 1$; for a correct initialization set $s(0) = 0$ and where there is no (semi)-repeat-free segmentation set $s(j) = -\infty$;
    \item[$ii$.]
for minimizing the maximum block length, set $\bigoplus = \min$ and $g(s(j'),j',j) = \max ( s(j'), L([j'+1,j]) ) = \max ( s(j'), j - j' )$; set $s(0) = 0$ and if there is no (semi)-repeat-free  segmentation set $s(j) = +\infty$.
\end{itemize}

Equi et al.\ studied the computation of semi-repeat-free segmentations optimizing for these two scores~\cite{Equietal21}. The algorithms they developed---and that we will improve in \Cref{sect:preprocessing,sect:minmaxlength}---are based on a common preprocessing of the valid semi-repeat-free segmentation ranges, based on the following observation.
\begin{myobservation}[Semi-repeat-free right extensions~\cite{Equietal21}]\label{obs:rightextensions}
Given a general $\msaij{1..m}{1..n}$, for any $x < y$ we say that segment $[x+1..y]$ is an extension of prefix $\msaij{1..m}{1..x}$. If segment $[x+1..y]$ is semi-repeat-free, then segment $[x+1..y']$ is semi-repeat-free for all $y < y' \le n$.
\end{myobservation}
Note that in the presence of gaps \Cref{obs:rightextensions} does not hold if we swap the semi-repeat-free notion with the repeat-free one, or if we swap the right extensions with the symmetrically defined left extensions.

In order to compute $s(j)$, \Cref{eq:score} considers all semi-repeat-free right extensions $[j'+1..j]$ ending at column $j$. Equi et al.\ discovered that the computation of values $s(j)$ can be done efficiently by considering that all semi-repeat-free right extensions $[j'+1..j]$ are prefixes of minimal (semi-repeat-free) right extensions $[j'+1..f(j')]$, with function $f$ defined as follows.
\begin{definition}[Minimal right extensions~\cite{Equietal21}]\label{def:rightextensions}
Given $\msaij{1..m}{1..n}$, for each $0 \le x \le n-1$ we define value $f(x)$ as the smallest integer greater than $x$ such that segment $[x+1..f(x)]$ is semi-repeat-free, or, in other words, $[x+1..f(x)]$ is the minimal (semi-repeat-free) right extension of prefix $\msaij{1..m}{1..x}$. If there is no semi-repeat-free extension, we define $f(x) = \infty$.
\end{definition}
Indeed, Equi et al.\ in \cite{Equietal21} developed an algorithm computing values $f(x)$ in time $O(mn\log m)$. Using only these values, described by a list of pairs $(x,f(x))$ sorted in increasing order by the second component, they developed two algorithms computing the score of an optimal semi-repeat-free segmentation: in time $O(n)$ for the maximum number of blocks score and in time $O(n \log\log n)$ for the maximum block length score. We will explain in detail how the latter works in \Cref{sect:minmaxlength}, as we will improve its run time to $O(n)$.

\section{Preprocessing the \msa in linear time}\label{sect:preprocessing}
In this section, we study the computation of the minimal right extensions $f(x)$, for $0 \le x \le n - 1$ (\Cref{def:rightextensions}).
Equi et al.\ in \cite{Equietal21} gave an $O(nm \log m)$-time algorithm for this \msa preprocessing using the following structure, built from the gaps-removed rows of the \msa.

\begin{definition}
Given $\msaij{1..m}{1..n}$ from alphabet $\Sigma \cup \lbrace \gap \rbrace$, we define \gst as the generalized suffix tree of the set of strings
$\lbrace \spell( \msaij{i}{1..n} ) \cdot \$_i : 1 \le i \le m \rbrace$, with $\$_1, \dots, \$_m$ $m$ new distinct terminator symbols not in $\Sigma$.%
\footnote{We added the $m$ new distinct terminators for simplicity, whereas Equi et al.\ used the suffix tree of the concatenation of all gaps-removed rows with a single new symbol $\$$ between each. The suffix tree of this string, if concatenated with a final suffix $\$\$$, is equivalent to \gst for our purposes.}
\end{definition}
An example of \gst is given in \Cref{fig:gstmsa}. From the suffix tree properties, it follows that for any gaps-removed row $\alpha_i \coloneqq \spell(\msaij{i}{1..n}) \$_i$:
each suffix $\alpha_i[x,\lvert \alpha_i \rvert]$ corresponds to a unique leaf $\ell_{i,x}$ of \gst and vice versa, with $1 \le x \le \lvert \alpha_i \rvert$;
each substring $\alpha_i[x,y]$ corresponds to an explicit or implicit node of \gst in the root-to-$\ell_{i,x}$ path; and each explicit or implicit node corresponds to one or more such substrings, uniquely identifiable thanks to the leaves covered by the node.
Also, note that \gst does not contain any information about the gap symbols of the \msa, as this information will be added back into the structure thanks to the set of leaves and nodes considered.

\begin{figure}[htp]
\centering
\includegraphics{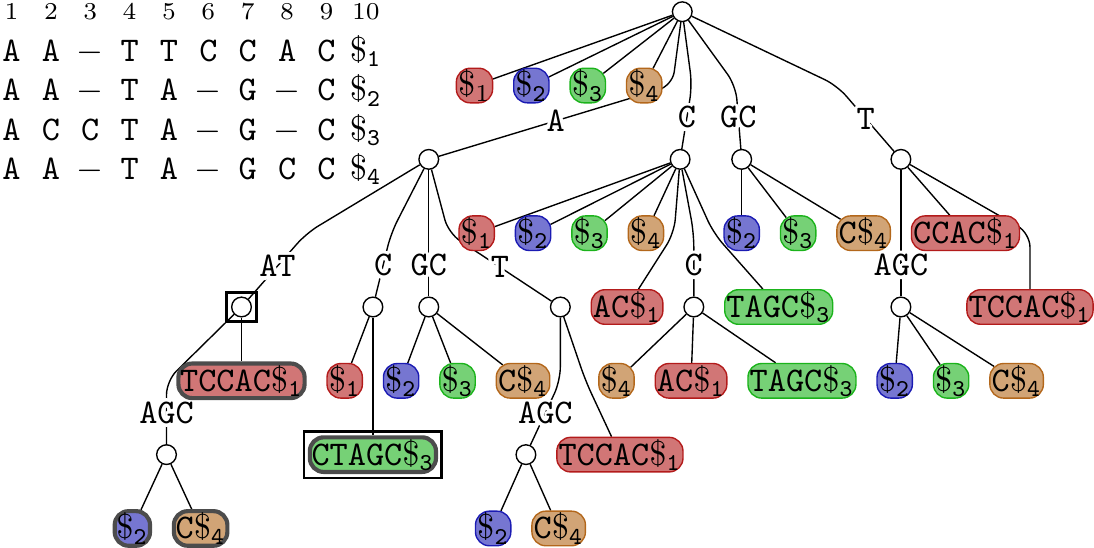}
\caption{Example of an $\msaij{1..4}{1..10}$ and its \gst, where each leaf is represented by the label of the corresponding leaf, which is colored according to the corresponding row. We have also highlighted  (with a grey outline) the leaves corresponding to suffixes $\spell(\msaij{i}{1..n})$ and its exclusive ancestors (with rectangles).
}\label{fig:gstmsa}
\end{figure}

In \Cref{sub:gst} we perform an analysis of \gst similar to that of Equi et al., showing that semi-repeat-free segments of the \msa correspond to a specific set of nodes of \gst covering exactly $m$ leaves. Then, in \Cref{sub:exclusiveancestorset}, we show that the novel resulting problem on the tree structure of \gst, that we call the \emph{exclusive ancestor set problem}, can be solved efficiently, resulting in an algorithm computing the minimal right extensions in linear time, described in \Cref{sub:minimalrightextensions}.

\subsection{Semi-repeat-free segments in the generalized suffix tree}\label{sub:gst}
\begin{definition}[Semi-repeat-free substrings]
Recall the definition of semi-repeat-free segment (\Cref{lem:characterization}). Given a substring $\msaij{i}{x..y}$ of $\msaij{1..m}{1..n}$ such that $\spell(\msaij{i}{x..y}) \in \Sigma^+$, we say that $\msaij{i}{x..y}$ is semi-repeat-free if, for all $1 \le i' \le m$, string $\spell(\msaij{i}{x..y})$ occurs in gaps-removed row $i'$ only at position $g(i',x)$ (or it does not occur at all).
\end{definition}
\begin{myobservation}\label{obs:rowsegmentf}
The following simple properties have been implicitly stated and exploited in~\cite{Equietal21}:
\begin{itemize}
    \item
segment $[x..y]$ is semi-repeat-free if and only if all substrings $\msaij{i}{x..y}$ are semi-repeat-free, for $1 \le i \le m$;
    \item
if $\msaij{i}{x..y}$ is semi-repeat-free, then $\msaij{i}{x..y'}$ is semi-repeat-free for all $y < y' \le n$;
    \item
let $f^i(x)$ be the smallest integer greater than $x$ such that substring $\msaij{i}{x+1..f^i(x)}$ is semi-repeat-free; it is easy to see that $f(x) = \max_{i=1}^{m} f^i(x)$.
\end{itemize}
\end{myobservation}

From the semi-repeat-free definition it follows that the strings of a semi-repeat-free segment occur at most once per row and only starting from a certain column range. This translates into a specific set of implicit or explicit nodes of \gst, with features that we can exploit to compute the minimal right extensions. The fact that we added a unique terminator symbol to each row is equivalent to the addition of an \msa column spelling $\$_1 \cdots \$_m$ at index $n+1$, meaning that a semi-repeat-free segment starting from column $x+1$ always exists and the minimal right extensions such that $f(x) = \infty$ becomes $f(x) = n + 1$.
\begin{lemma}\label{lem:gst}
Given $m$ row substrings $\msaij{i}{x..y_i}$ of $\msaij{1..m}{1..n}$ such that $\spell(\msaij{i}{x..y_i}) \in \Sigma^+$ for all $1 \le i \le m$,
let $W = \lbrace w_1, \dots, w_k \rbrace$ be the set of implicit or explicit nodes of \gst corresponding to strings
$\lbrace \spell( \msaij{i}{x..y_i} : 1 \le i \le m \rbrace$.
Then
    $\msaij{i}{x..y_i}$ is semi-repeat-free for all $1 \le i \le m$
    if and only if
    $W$ covers exactly $m$ leaves in \gst.
\end{lemma}
\begin{proof}
By construction of \gst, $W$ covers the $m$ leaves $\ell_{1,z_i}, \dots, \ell_{m,z_m}$, with $z_i = g(i,x)$, so we only need to prove that if some $\msaij{i}{x..y_i}$ is not semi-repeat-free, or \emph{invalid}, then $W$ covers more than $m$ leaves, and vice versa.

($\Leftarrow$) Let $\msaij{i}{x..y_i}$ be invalid, i.e.\ $\spell(\msaij{i}{x..y_i})$ occurs in $\alpha_{i'}$ at some position $\hat{z}$ other than $z_{i'}$, for some row $1 \le i' \le m$. Then the node of \gst corresponding to string $\spell(\msaij{i}{x..y_i})$ covers leaf $\ell_{i',\hat{z}} \neq \ell_{i',z_{i'}}$, thus $W$ covers more than $m$ leaves.

($\Rightarrow$) Let $\ell_{i',\hat{z}}$ be a leaf of \gst other than leaves $\ell_{1,z_1}, \dots, \ell_{m,z_m}$ and covered by some node $w \in W$. By construction, $w$ corresponds to $\spell( \msaij{i}{x..y_i} )$ for some $1 \le i \le m$, so we have that $\spell( \msaij{i}{x..y_i} )$ occurs in $\alpha_{i'}$ at some position other than $g(i',x)$, since $\ell_{i',\hat{z}} \neq \ell_{i',z_{i'}}$. Thus, $\msaij{i'}{x..y_i}$ is invalid.
\end{proof}
Note that the correctness of \Cref{lem:gst} does not hold if we swap the semi-repeat-free notion with the repeat-free one.

\Cref{lem:gst}, combined with \Cref{obs:rowsegmentf}, implies that the problem of computing values $f^i(x)$ for all $i \in [1..m]$ can be solved by analyzing the tree structure of \gst against the \msa suffixes.
Indeed, let $L_x \coloneqq \lbrace \ell_{i,z_i} : 1 \le i \le m, z_i = g(i,x+1) \rbrace$ be the leaves of \gst corresponding to the suffixes $\spell(\msaij{i}{x+1..n})$.
For each row $1 \le i \le m$, the first semi-repeat-free prefix of $\spell(\msaij{i}{x+1..n})$ corresponds to the first implicit or explicit node $v$ of \gst in the root-to-$\ell_{i,z_i}$ path such that $v$ covers only leaves in $L_x$.
The fact that \gst is a compacted trie is not an issue: the parent of $v$ in the suffix trie is branching, since it covers more leaves than $v$, so the first explicit node of \gst in the root-to-$\ell_{i,z_i}$ path covering only leaves in $L_x$ is the first explicit descendant $w$ of $v$, thus we can identify $v$ by finding $w$.
Finally, $f^i(x)$ is computed by retrieving the first column index $y$ such that $\spell(\msaij{i}{x+1..y}) = \sstring(\parent(w)) \cdot \cchar(w)$, where $\sstring(u)$ is the concatenation of edge labels of the root-to-$u$ path, and $\cchar(u)$ is the first symbol of the edge label from $\parent(u)$ to $u$. In other words, $y$ corresponds to the $k$-th non-gap symbol of \msa row $i$, with $k = \mathrm{rank}(\msaij{i}{1..n}, x) + \stringdepth(\parent(w)) + 1$, where $\mathrm{rank}(\msaij{i}{1..n}, x)$ is the number of gap symbols in $\msaij{i}{1..x}$ and $\stringdepth(u) = \lvert \sstring(u) \rvert$.
For example, in \Cref{fig:gstmsa} the leaves of $L_0$ have been marked and so have the shallowest ancestors covering only leaves in $L_0$.

\subsection{Exclusive ancestor set}\label{sub:exclusiveancestorset}
The results of the previous section show that we can compute the minimal right extensions by solving multiple instances of the following problem on the tree structure of \gst.
\begin{problem}[Exclusive ancestor set]
Let $T = (V,E,\mathrm{root})$ be a rooted ordered tree, with $L^T \subseteq V$ the set of its leaves. Given $T$ and a subset of leaves $L \subseteq L^T$, find the minimal set $W$ of exclusive ancestors of $L$ in $T$, i.e.\ the minimal set $W \subseteq V$ such that $W$ covers all leaves in $L$ and only leaves in $L$.
Can $T$ be preprocessed to support the efficient solving of multiple instances of the problem?
\end{problem}

As is the case for \gst, we can assume that each internal node of $T$ has at least two children, otherwise, a linear-time processing of $T$ can be employed to compact its unary paths. Indeed, after a linear-time preprocessing of $T$, any instance of exclusive ancestor set can be solved in time $O(\lvert L \rvert)$ by a careful traversal of the tree with the following procedure, that we describe informally:
\begin{enumerate}
    \item partition $L$ in $k$ maximal sets $L_1$, \dots, $L_k$ of leaves contiguous in the ordered traversal of $T$, to be processed independently (if two leaves belong to different contiguous sets, any common ancestor cannot be part of the solution);
    \item for each $L_i$, with $1 \le i \le k $, start from the leftmost leaf $\ell_i$ and ascend in the tree until the closest ancestor of $\ell_i$ that covers some leaf not in $L_i$;
    \item upon failure in step 2., add the last safe ancestor to the solution $W$ and if there are still uncovered leaves in $L_i$ repeat steps 2.\ and 3.\ starting from the leftmost uncovered leaf.
\end{enumerate}
An example of the procedure is shown in \Cref{fig:exclusiveancestorset}. The failure condition of step 2.\ can be evaluated by checking if both the leftmost leaf and rightmost leaf in the subtree of the candidate replacement are still in set $L_i$, and step 2.\ always terminates if we assume that $L$ is a non-trivial instance such that $L \subset L^T$ so the root of $T$ is not the solution to the problem.
\begin{figure}[htp]
\centering
\includegraphics{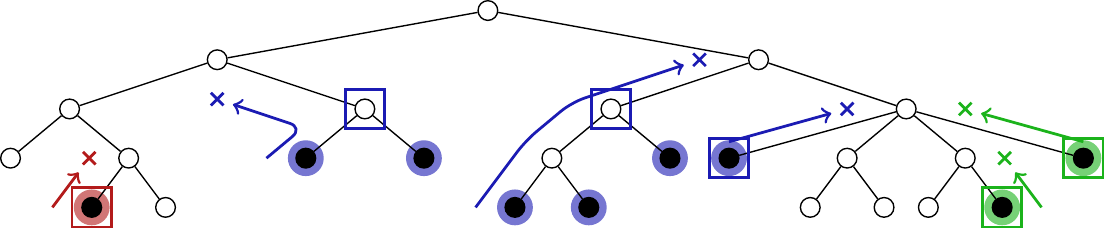}
\caption{Example of an instance of exclusive ancestor set, where the set of leaves $L$ correspond to the leaves filled in black: the algorithm partitions $L$ into sets of contiguous leaves (shown as red, blue, and green leaves), and for each set it finds the exclusive ancestors (marked with rectangles). Each arrow shows the ascent of step 2.\ up the tree until the node corresponding to the failure condition, marked with a cross.}\label{fig:exclusiveancestorset}
\end{figure}

Assuming the leaves of $T$ are sorted, step 1.\ can be implemented efficiently: we can partition $L$ into sets of contiguous leaves by coloring leaves in $L$ and finding all the leaves with the preceding leaf not in $L$.
We can easily preprocess $T$ to support the needed operations in constant time, leading to a time complexity of $O(\lvert L \rvert)$, since any forest built on top of leaves $L$ has $O(\lvert L \rvert)$ nodes.
\begin{lemma}\label{lem:exclusiveancestorset}
The exclusive ancestor set problem on a rooted ordered tree $T = (V,E,\mathrm{root})$ and a subset $L$ of its leaves can be solved in time O$(\lvert L \rvert)$, after a $O(\lvert T \rvert)$-time preprocessing to support operations $v.\leftmostleaf$, $v.\rightmostleaf$ on any node $v \in V$ and operations $\ell.\prevleaf$, $\ell.\nextleaf$, and the binary coloring of any leaf $\ell \in L^T$ in constant time.
\end{lemma}

\subsection{Computing the minimal right extensions}\label{sub:minimalrightextensions}
Returning to the problem of computing values $f(x)$, the representation of \gst needs to support the operations on its tree structure described by \Cref{lem:exclusiveancestorset} plus operations $v.\stringdepth$, returning the length of the string corresponding to the root-to-$v$ path in \gst of an explicit node $v$, and $\ell.\suffixlink$, implementing the suffix links of the leaves.
The final algorithm, described in \Cref{alg:minimalrightextensions}, computes leaf sets $L_0$, $L_1$, \dots, $L_{n-1}$ corresponding to the \msa suffixes starting at column $1, 2, \dots, n$, respectively, and for each $L_x$:
\begin{enumerate}
    \item it marks the leaves in $L_x$ and partitions them in sets of contiguous leaves, by finding all their left boundaries $\ell$ such that $\ell.\prevleaf$ is not marked;
    \item it solves the exclusive ancestor set problem on each set of contiguous leaves and whenever it finds an exclusive ancestor, covering leaves $\ell_{i_1}, \dots, \ell_{i_k}$, it computes values $f^i(x)$ for $i \in \lbrace i_1, \dots, i_k \rbrace$ (see the conclusion of \Cref{sub:gst});
    \item after processing all leaves, it finally computes $f(x) = \max_{i=1}^{m} f^i(x)$ and transforms $L_x$ into $L_{x+1}$ by taking the suffix links of only leaves $\ell_{i}$ such that $\msaij{i}{x+1} \neq \gap$.
\end{enumerate}

\begin{algorithm}
\SetKw{KwOutput}{output}
\SetKwRepeat{Do}{do}{while}
\KwIn{$\msaij{1..m}{1..n}$ from alphabet $\Sigma \cup \lbrace \mathtt{-} \rbrace$}
\KwOut{Pairs $(x,f(x))$ for $x = 0, \dots, n-1$}
Preprocess the \msa rows to support select and rank queries\;
Build \gst, the generalized suffix tree of $\lbrace \spell(\msaij{i}{1..n}) \$_i : 1 \le i \le m \rbrace$\;
Initialize $\mathtt{L}[i]$ as the leaf of \gst corresponding to $\spell(\msaij{i}{1..n})$\;
\For{$x \gets 0 \;\KwTo\; n-1$}{
    \For(\tcp*[f]{Mark each leaf of $L_x$}){$i \gets 1 \;\KwTo\; m$}{
        $\mathtt{L}[i].\mathrm{marked} \gets \atrue$\;
    }
    \tcp{Process each set of contiguous leaves}
    \For{$\mathrm{i} \in [1..m] : \mathtt{L}[i].\mathrm{prevleaf}.\mathrm{marked} = \afalse$}{
        $\mathrm{lb} \gets \mathtt{L}[i]$\tcc*[r]{Find left and right boundaries}
        $\mathrm{rb} \gets \mathrm{lb}$\;
        \While{$\mathrm{rb}.\mathrm{nextleaf}.\mathrm{marked} = \atrue$}{
            $\mathrm{rb} \gets \mathrm{rb}.\mathrm{nextleaf}$\;
        }
        $w \gets \mathrm{lb}; \; \mathrm{lleaf} \gets \mathrm{lb}; \;  \mathrm{rleaf} \gets \mathrm{lb}$\tcc*[r]{Find the exclusive ancestors}
        \While{$\mathrm{rleaf} \le \mathrm{rb}$}{
            $w' \gets w.\mathrm{parent}; \; \mathrm{lleaf}' \gets w'\!.\mathrm{leftmostleaf}; \; \mathrm{rleaf}' \gets w'\!.\mathrm{rightmostleaf}$\;
            \uIf(\tcp*[f]{$w'$ is a correct replacement}){$\mathrm{lb} \le \mathrm{lleaf}' \,\wedge\, \mathrm{rleaf}' \le \mathrm{rb}$}{
                $w \gets w'; \; \mathrm{lleaf} \gets \mathrm{lleaf}'; \; \mathrm{rleaf} \gets \mathrm{rleaf}'$\;
            }
            \Else(\tcp*[f]{$w'$ fails so $w$ is an exclusive ancestor}){
                $\ell \gets \mathrm{lleaf}$\;
                \While{$\ell \le \mathrm{rleaf}$}{
                    $i' \gets \ell.\mathrm{getrow}$\;
                    $g \gets (w.\mathrm{parent}).\mathrm{stringdepth} + 1$\;
                    $\mathtt{f}[i'] \gets \msaij{i'}{1..n}.\mathrm{select}( \msaij{i'}{1..n}.\mathrm{rank}(x) + g )$\;
                    $\ell \gets \ell.\mathrm{nextleaf}$\;
                }
                $w \gets \ell; \; \mathrm{lleaf} \gets \ell; \;  \mathrm{rleaf} \gets \ell$\;
            }
        }
    }
    \KwOutput{$(x,\max_{i=1}^{m}\mathtt{f}[i])$}\;
    \For(\tcp*[f]{Cleanup and compute $L_{x+1}$}){$i \gets 1 \;\KwTo\; m$}{
        $\mathtt{leaf}[i].\mathrm{marked} \gets \afalse$\;
        \If{$\msaij{i}{x+1} \neq \mathtt{-}$}{
            $\mathtt{leaf}[i] \gets \mathtt{leaf}[i].\suffixlink$\;
        }
    }
}
\caption{Algorithm computing the minimal right extensions $f(x)$, for $0 \le x \le n - 1$. 
For simplicity, we use the total order $\le$ to check the failure condition of the exclusive ancestor set problem, but it is easy to see that $\le$ can be replaced by an additional binary coloring of the leaves, to check if a leaf is contained in a set of contiguous leaves of interest.}\label{alg:minimalrightextensions}
\end{algorithm}

\begin{theorem}\label{teo:minimalrightextensions}
Given $\msaij{1..m}{1..n}$, we can compute the minimal right extensions $f(x)$ for $0 \le x \le n-1$ in time $O(mn)$.
\end{theorem}
\begin{proof}
The correctness is given by \Cref{obs:rowsegmentf} and \Cref{lem:gst,lem:exclusiveancestorset}. The construction of \gst is equivalent to building a string of length smaller or equal than $(m+1) \cdot n$ and its suffix tree: a suffix tree supporting the required operations in constant time can be constructed in $O(mn)$ time \cite{Farach97}, assuming an integer alphabet of size smaller or equal than $(m+1) \cdot n$. Plus, we can preprocess the \msa rows to answer in constant time rank and select queries on the position of gap and non-gap symbols~\cite{jacobson1988succinct}. Finally, the computation of each $f(x)$ takes time $O(\lvert L_x \rvert + m) = O(m)$, so $O(mn)$ time in total.
\end{proof}

\begin{corollary}
Given $\msaij{1..m}{1..n}$ from $\Sigma \cup \lbrace \gap \rbrace$, with $\Sigma$ an integer alphabet of size $O(mn)$, the construction of an optimal semi-repeat-free segmentation minimizing the maximum number of blocks can be done in time $O(mn)$.
\end{corollary}
\begin{proof}
The algorithm \cite[Algorithm 1]{Equietal21} by Equi et al.\ solves the problem in $O(n)$ time, assuming it is given the minimal right extensions $(x,f(x))$ ordered by the second component, which we can now compute and sort in time $O(mn)$ thanks to \Cref{teo:minimalrightextensions}.
\end{proof}

\begin{myobservation}
Note that in the setting with gaps the correctness of \Cref{lem:gst} and \Cref{alg:minimalrightextensions} does not hold if we swap the notion of semi-repeat-freeness with the one of repeat-freeness, because \gst does not explicitly contain the information on what are the ending columns of the strings considered.
\end{myobservation}

\section{Minimizing the maximum block length}\label{sect:minmaxlength}
The improvement on the computation of the minimal right extensions in the case of general \msas from $O(nm \log m)$ to $O(nm)$ gives us the motivation to improve the $O(n \log \log n)$-time algorithm of Equi et al.~\cite[Algorithm 2]{Equietal21} for an optimal semi-repeat-free segmentation minimizing the maximum block length.
As mentioned in \Cref{sub:efgconstruction}, we can compute $s(j)$ by processing the recursive solutions corresponding to all right extensions $(x,f(x))$ with $f(x) \le j$. For the maximum block length there are two types of recursion for an optimal solution of $\msaij{1..m}{1..j'}$ using semi-repeat-free $[x+1..j']$ as its last segment:
\begin{description}
    \item[non-leader recursion:] if $j' \le x + s(x)$ then the score of $s(j')$ is equal to $s(x)$, because the length of segment $[x+1..j']$ is less than or equal to $s(x)$; in this case, we say that $[x+1..j']$ is a \emph{non-leader} segment;
    \item[leader recursion:] otherwise, if $j' > x + s(x)$, we say that $[x+1..j']$ is a \emph{leader segment}, since it gives score $j' - x$ to an optimal solution constrained to use it as its last segment.
\end{description}
\begin{figure}[htp]
\centering
\includegraphics{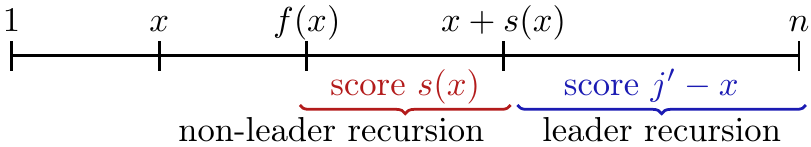}
\caption{Scheme for the score of an optimal semi-repeat-free segmentation of $\msaij{1..m}{1..j'}$ constrained to use $[x+1..j']$ as its last segment, with $(x,f(x))$ the minimal right extension of $\msaij{1..m}{1..x}$.}
\end{figure}
Note that if $x + s(x) < f(x)$ then the first type of recursion does not happen for $(x,f(x))$. Then, it is easy to see that
\begin{equation}
    s(j) =
    \min \Bigg(
        \min_{\substack{(x,f(x)) : \\ f(x) \le j \le x + s(x)}} s(x),
        \min_{\substack{(x,f(x)) : \\ j > f(x) \;\wedge\; j > x + s(x)}} j - x
    \Bigg)
\end{equation}
so Equi et al.\ correctly solve the problem by keeping track of the two types of recursions with two data structures: the first keeps track of ranges $[f(x)..x + s(x)]$ with score $s(x)$, the second tracks ranges $[x + s(x) + 1..n]$ where the leader recursion has to be used, reaching a global time of $O(n \log \log n)$.

Instead, we can reach a linear time complexity using simpler data structures, thanks to the following observations:
\begin{itemize}
    \item
the data structure for the leader recursion can be replaced by a single variable $S$ holding value $\min \lbrace j - x : j > f(x) \wedge j > x + s(x) \rbrace$, so that $S$ is the best score of a segmentation ending with a leader segment $[x+1..j]$;
    \item
for the non-leader recursion, we can swap the structure of Equi et al.\ with an equivalent array $\mathtt{C}[1..n]$ such that $\mathtt{C}[i]$ counts the number of available solutions with score $i$ using the non-leader recursion so that a variable $I = \min \lbrace i : \mathtt{C}[i] > 0 \rbrace$ is equal to the best score of a segmentation ending with a non-leader segment $[x+1..j]$.
\end{itemize}
The final and crucial observation is that the two types of recursion are closely related: when $[x+1..j]$ goes from being a non-leader segment to being a leader one, i.e.\ $j = x + s(x) + 1$, we decrease $\mathtt{C}[s(x)]$ by one and update $S$ with value $s(x) + 1 = j - x$ if needed. Thus, when the best score of $\mathtt{C}[1..n]$ is removed in this way, we do not need to update $I$ to $\min \lbrace  i : \mathtt{C}[i] > 0 \rbrace$, but it is sufficient to increment $I$ by $1$ to make sure that $s(j) = \min (I,S)$, unless other updates of $\mathtt{C}$ and $S$ result in a better score.
Thus, in order to compute $s(j + 1)$ variable $I$ needs to be incremented by one if $\mathtt{C}[i] = 0$ and variable $S$ also needs to be incremented by one, and the algorithm performs a constant number of operations on $\mathtt{C}$ for each $s(j)$ other than the $O(n)$ global time spent in managing the minimal right extensions.

\begin{algorithm}
\KwIn{Minimal right extensions $(x_1, f(x_1)), \dots, (x_n, f(x_n))$ sorted from smallest to largest order by the second component}
\KwOut{Score of an optimal semi-repeat-free segmentation minimizing the maximum block length}
Initialize array $\mathtt{C}[1..n]$ with values in $[0..n]$ and set all values to $0$\;
Initialize array $\mathtt{L}[1..n]$ as empty linked-lists with values pointers to $(x, f(x))$\;
$\mathtt{minmaxlength}[0] \gets 0$\;
$y \gets 1; \; I \gets 1; \; S \gets \infty$\;
\For{$j \gets 1 \;\KwTo\; n$}{
    \While(\tcp*[f]{Process minimal right extensions}){$j = f(x_y)$}{
        \uIf(\tcp*[f]{Non-leader recursion}){$j \le x_y + \mathtt{minmaxlength}[x_y]$}{
            $\mathtt{C}[\mathtt{minmaxlength}[x_y]] \gets \mathtt{C}[\mathtt{minmaxlength}[x_y]] + 1$\;
            $I.\mathrm{Update} ( \mathtt{minmaxlength}[x_y] )$\;
            $\mathtt{L}[x_y + \mathtt{minmaxlength}[x_y] + 1].\mathrm{add}\big( (x_y, f(x_y)) \big)$\;
        }
        \Else(\tcp*[f]{Leader recursion}){
            $S.\mathrm{Update}( j - x_y )$\;
        }
        $y \gets y + 1$\;
    }
    \For(\tcp*[f]{Update non-leader recursions into leader ones}){$(x,f(x)) \in \mathtt{L}[j]$}{
        $\mathtt{C}[\mathtt{minmaxlength}[x]] \gets \mathtt{C}[\mathtt{minmaxlength}[x]] - 1$\;
        $S.\mathrm{Update}(j - x)$ \tcp*{$j - x = s(x) + 1$}
    }

    \uIf{$\mathtt{C}[I] > 0$}{
        $\mathtt{minmaxlength}[j] \gets \min(I, S)$\;
    }
    \Else{
        $\mathtt{minmaxlength}[j] \gets S$\;
    }

    $S \gets S + 1$ \tcp*{Update the data structures for next iteration}
    \If{$\mathtt{C}[I] = 0$}{
        $I \gets I + 1$\;
    }
}
\Return{$\mathtt{minmaxlength}[n]$}\;
\caption{Main algorithm to find the optimal score of a semi-repeat-free segmentation minimizing the maximum block length. Operations $I.\mathrm{Update}$ and $S.\mathrm{Update}$ replace the current value of the variable with the input, if the input is smaller.}\label{alg:minmaxlength}
\end{algorithm}

\begin{theorem}
Given the minimal right extensions $(x,f(x))$ of $\msaij{1..m}{1..n}$, we can compute in time $O(n)$ the score of an optimal semi-repeat-free segmentation minimizing the maximum block length.
\end{theorem}
\begin{proof}
The correctness of the algorithm, implemented in \Cref{alg:minmaxlength}, follows from that of \cite[Algorithm 2]{Equietal21} and from the fact that when $\mathtt{C}[I] = 0$ we have that $\mathtt{C}[j'] = 0$ for $1 \le j' \le I$ and $S \le I + 1$. Similarly, the processing of minimal right extensions $(x,f(x))$ and the dynamic management of intervals $[f(x), s(x) + j']$ takes time $O(n)$ in total, thus the algorithm takes linear time.
\end{proof}

\newpage
Combined with the linear-time computation of minimal right extensions of \Cref{teo:minimalrightextensions} and with the fact that we can sort them in linear time, we get our second main result.
\begin{corollary}
Given $\msaij{1..m}{1..n}$ from $\Sigma \cup \lbrace \gap \rbrace$, with $\Sigma$ an integer alphabet of size $O(mn)$, the construction of an optimal semi-repeat-free segmentation minimizing the maximum block length can be done in time $O(mn)$.
\end{corollary}

\bibliographystyle{splncs04}
\bibliography{biblio}

\end{document}